\documentclass[12pt]{article}
\usepackage{amsthm,amsfonts,amsmath,amssymb,amscd, accents, mathrsfs}
\usepackage{authblk}
\usepackage{hyperref}

                                \usepackage{verbatim}
\swapnumbers
                              
\theoremstyle{plain}

\usepackage{enumitem}
\numberwithin{equation}{section}
\newtheorem{theorem}{Theorem}[section]

\newtheorem{lemma}[theorem]{Lemma}

\newtheorem{definition}[theorem]{Definition}
\newtheorem{notation}[theorem]{Notation}

\theoremstyle{remark}
\newtheorem{remark}[theorem]{Remark}
\numberwithin{equation}{section}

\newcommand{\cS}{{\mathcal S}}

\newcommand{\cH}{{\mathcal H}}

\newcommand{\cC}{{\mathcal C}}

\newcommand{\cI}{{\mathcal I}}

\newcommand{\cQ}{{\mathcal Q}}

\newcommand{\cD}{{\mathcal D}}

\newcommand{\ket}[1]{\left\vert #1\right\rangle}
\newcommand{\bra}[1]{\left\langle #1\right\vert}

\def\idty{{\mathchoice {\mathrm{1\mskip-4mu l}} {\mathrm{1\mskip-4mu l}} %
{\mathrm{1\mskip-4.5mu l}} {\mathrm{1\mskip-5mu l}}}}

\newcommand{\Tr}{\mathrm{Tr}}

\renewcommand{\vec}[1]{\boldsymbol{#1}}

\newcommand{\be}{\begin{equation}}
\newcommand{\ee}{\end{equation}}
\newcommand{\bea}{\begin{eqnarray}}
\newcommand{\eea}{\end{eqnarray}}
\newcommand{\beann}{\begin{eqnarray*}}
\newcommand{\eeann}{\end{eqnarray*}}

\usepackage{color}

\begin{document}

\title{Quantum coherence, discord and correlation measures based on Tsallis relative entropy}
\author{Anna Vershynina}
\affil{\small{Department of Mathematics, Philip Guthrie Hoffman Hall, University of Houston, 
3551 Cullen Blvd., Houston, TX 77204-3008, USA}}
\renewcommand\Authands{ and }
\renewcommand\Affilfont{\itshape\small}

\date{\today}

\maketitle

\begin{abstract} Several ways have been proposed in the literature to define a coherence measure based on Tsallis relative entropy. One of them is defined as a distance between a state and a set of incoherent states with Tsallis relative entropy taken as a distance measure. Unfortunately, this measure does not satisfy the required strong monotonicity, but a modification of this coherence has been proposed that does. We introduce three new Tsallis coherence measures coming from a more general definition that also satisfy the strong monotonicity, and compare all five definitions between each other. Using three coherence measures that we discuss, one can also define a discord. Two of these have been used in the literature, and another one is new. We also discuss two correlation measures based on Tsallis relative entropy. We provide explicit expressions for all three discord and two correlation measure on pure states.  Lastly, we provide tight upper and lower bounds on two discord and correlations measures on any quantum state, with the condition for equality.
\end{abstract}

\section{Introduction}
Coherence is the fundamental property of quantum systems, that is used in thermodynamics \cite{A14, C15, L15}, transport theory \cite{RM09, WM13}, and quantum optics \cite{G63, SZ97}, among few applications. Recently, the problems involving coherence included quantification of coherence \cite{BC14, PC16, RPL16, R16, SX15, YZ16}, distribution \cite{RPJ16}, entanglement \cite{CH16, SS15}, operational resource theory \cite{CG16, CH16, DBG15, WY16},  correlations \cite{HH18, MY16, TK16}, with only a few references mentioned in each. See \cite{ZY18} for a more detailed review.

Here we focus on the problem of quantification of coherence. A number of ways has been proposed as a coherence measure, but only a few satisfy all necessary criteria \cite{BC14, ZY18, Zetal17}. A ``good" measure of coherence should satisfy a list of intuitive criteria, which include a strong monotonicity property - a monotonicity under selective incoherent CPTP maps. Measures that satisfy the strong monotonicty that have been introduced up to date, are based on $l_1$, relative entropy, Tsallis entropy, and real symmetric concave functions on a probability simplex.  Tsallis coherence measure can be defined several ways. Here we discuss five different ways to define a Tsallis coherence measure, two of them have been used in the literature before \cite{R16, ZY18}, and we also consider three  particular cases of a more broad definition \cite{Zetal17}.

The first coherence measure $\cC_q^I$ discussed here, was defined in \cite{R16}. The measure is defined as a convex roof, and the explicit expression for it was also given, see (\ref{eq:Ts-1}) for details. Authors showed that the measure satisfies several properties expected for a coherence measure, but not the strong monotonicity under incoherent selective measurements.

To rectify the situation, a new coherence measure, $\cC_q^{II}$ was defined in \cite{ZY18}, see (\ref{eq:Ts-2}). This measure does satisfy the strong monotonicity under incoherent selective measurements along with all previous ones. Now, one can argue for two sides here. On one hand, $\cC_q^I$ has a good intuition behind it, as it is defined as the distance between a given state and the set of incoherence states, with the relative Tsallis entropy as a a distance measure. The new measure is a more artificial looking one, constructed for a specific purpose of satisfying the list of given properties. On the other hand, $\cC_q^{II}$ is easy to compute, and it does satisfies the list of properties a good coherence measure should.

A class of coherence measures was defined in \cite{Zetal17} based on real symmetric concave functions on a probability simplex. For any such function one can explicitly define a coherence measure on pure states, and as a convex roof for  mixed states. We use $\cC_q^I$ and $\cC_q^{II}$ to define coherence on pure states, and use the convex roof construction to define two coherence measures on mixed states. Moreover, we define another coherence measure based on a familiar concave function.

Section \ref{sec:coherence} is used to compare all five measures between each other.

One can define a discord in different ways. Here we consider three definitions of the discord based on the first three coherence measures. This is not a complete list of possible definitions, see, for example,  \cite{J13} for another definition.

In \cite{L19}, the measure of quantum correlations $\cQ_q$ was defined as  a minimum Tsallis relative entropy over a set of classical-quantum states, a bigger set than a set of classically-correlated states. The authors also defined a new correlation measure based on the second definition of the coherence measure. The correlation measure in \cite{L19} is non-greater then discord, and we show that they coincide on pure states.

We present a tight lower bound for two discord and correlation measures, and show when it is saturated.
The  upper bound for the coherence measure $\cQ_q$ was derived in \cite{R16}, which is the upper bound for the discord as well. Here we show that the upper bound is tight for the discord as well, and present a tight lower bound for the second discord and correlation measures.
In summary, we present three results:
\begin{itemize}
\item comparison between five coherence measures, Theorems \ref{thm:I-III}, \ref{thm:I-IV}, \ref{thm:II-V};
\item explicit expression for three discord  and two correlation measures on pure states, Theorem \ref{lemma:1};
\item a tight lower bound on two discord and correlation measures, Theorem \ref{thm:lower};
\item a tight upper bound on two discord and correlation measures, Theorem \ref{thm:upper}.
\end{itemize}

\section{Preliminaries}
\subsection{Coherence}

Let $\cH$ be a $d$-dimensional Hilbert space. Let us fix a basis $\{\ket{j}\}_{j=1}^d$ of vectors in $\cH$.
\begin{definition} A state $\delta$ is called {\it incoherent} if it can be represented as follows
$$\delta=\sum_j \delta_j\ket{j}\bra{j}\ . $$
\end{definition}

In the case of multipartite system, let the corresponding Hilbert spaces be denoted as $\{\cH_k\}_1^N$. Suppose each Hilbert space is of dimension $d_k$ and has a fixed basis $\{\ket{j}\}_{j=1}^{d_k}$.

\begin{definition} A state $\delta$ on $\cH=\bigotimes_k \cH_k$ is called {\it incoherent} if it can be represented as follows
$$\delta=\sum_{\vec{j}} \delta_{\vec{j}}\ket{\vec{j}}\bra{\vec{j}}\ , $$
where $\vec{j}=(j_1, \dots, j_N)$ with $j_k=1, \dots, d_k$ is a vector of indices, and $\ket{\vec{j}}=\bigotimes_k \ket{j_k}=\ket{{j_1}}\otimes\cdots\otimes\ket{{j_N}}.$
\end{definition}

\begin{notation}
Denote the set of {\bf incoherent states} for a fixed basis $\{\ket{i}\}_i$ as $$\cI=\{\rho=\sum_jp_j\ket{j}\bra{j}\}\ .$$
Denote the set of {\bf separable states} as $$\cS=\{\rho_{AB}=\sum_{ij}p_{i,j}\sigma_i^A\otimes\sigma_j^B\ |\  \text{for any states } \sigma_i, \sigma_j\}\ .$$
Denote the set of {\bf classically correlated states} as $$\cC\cC=\{\rho_{AB}=\sum_{i,j}p_{i,j}\ket{ij}\bra{ij}=\sum_{i,j}p_{i,j}\ket{i}\bra{i}_A\otimes\ket{j}\bra{j}_B\ | \ \text{for any bases } \{\ket{i}_A\},  \{\ket{j}_B\}\}\ .$$
\end{notation}

A CPTP quantum channel is categorized into the following two classes.

\begin{definition} A CPTP map $\Phi$ with the following Kraus operators
$$\Phi(\rho)=\sum_n K_n \rho K_n^*\ , $$
is called {\bf the non-selective incoherent CPTP (ICPTP)} when the Kraus operators satisfy
$$K_n \cI K_n^*\subset \cI,\ \text{for all }n \ , $$
besides the regular completeness relation $\sum_n K_n^*K_n=\idty$.
\end{definition}

\begin{definition} A CPTP map $\Phi$ with the following Kraus operators
$$\Phi(\rho)=\sum_n K_n \rho K_n^*\ , $$
is called {\bf the selective ICPTP} when the Kraus operators satisfy
$$K_n \cI K_n^*\subset \cI,\ \text{for all }n \ , $$
besides the regular completeness relation $\sum_n K_n^*K_n=\idty$. Additionally, we record the outcomes of each measurement
$$\rho_n=\frac{K_n\rho K_n^*}{p_n},\ \ p_n=\Tr K_n\rho K_n^*\ . $$
\end{definition}

Any reasonable measure of coherence $\cC(\rho)$ should satisfy the following conditions
\begin{itemize}
\item (C1) $\cC(\rho)=0$ if and only if $\rho\in\cI$;
\item (C2) Monotonicity under non-selective ICPTP maps (monotonicity)
$$\cC(\rho)\geq \cC(\Phi(\rho))\ ; $$
\item (C3) Monotonicity under selective ICPTP maps  (strong monotonicity)
$$\cC(\rho)\geq \sum_n p_n \cC(\rho_n)\ , $$
where $p_n$ and $\rho_n$ are the outcomes and post-measurement states defined above;
\item (C4) Convexity, 
$$\sum_n p_n \cC(\rho_n)\geq \cC\left(\sum_n p_n\rho_n\right)\ , $$
for any sets of states $\{\rho_n\}$ and any probability distribution $\{p_n\}$.
\end{itemize}

These properties are parallel with the entanglement measure theory, where the average entanglement is not increased under the local operations and classical communication (LOCC). Notice that coherence measures that satisfy conditions (C3) and (C4) also satisfies condition (C2). 

Any distance measure $D$ between two quantum states, can induce a potential candidate for coherence. The distance-based coherence measure is defined as follows \cite{BC14}.

\begin{definition}
$$\cC_D(\rho):=\min_{\delta\in\cI}D(\rho, \delta)\ ,$$
i.s. the minimal distance between the state $\rho$ and the set of incoherent states $\cI$ measured by the distance $D$.
\end{definition}

\begin{itemize}
\item (C1) is satisfied whenever $D(\rho, \delta)=0$ iff $\rho=\delta$.
\item (C2) is satisfied whenever $D$ is contracting under CPTP maps, i.e. $D(\rho, \sigma)\geq D(\Phi(\rho), \Phi(\sigma))$.
\item (C4) is satisfied whenever $D$ is jointly convex.
\end{itemize}

\subsection{Tsallis entropy and relative entropies}
\begin{definition}
For $0<q\neq 1$, the {\bf Tsallis entropy} is defined as
$$S_q(\rho)=\frac{1}{q-1}(1-\Tr\rho^q)\ . $$
\end{definition}
In the limit $q\rightarrow 1$, the Tsallis entropy becomes von Neumann entropy
$$\lim_{q\rightarrow 1}S_q(\rho)=S(\rho):=-\Tr(\rho\log\rho)\ . $$

\begin{definition}
For $0<q\neq 1$, the {\bf Tsallis relative entropy} is defined as
$$D_q(\rho\|\sigma)=\frac{1}{q-1}(\Tr(\rho^q\sigma^{1-q})-1)\ . $$
\end{definition}

The Tsallis relative entropy satisfies the following properties:
\begin{itemize}
\item The Tsallis relative entropy is zero if and only if $\rho=\sigma$.
\item For $q\in(0,2]$, in \cite{HM11, R16}, it was shown that  the Tsallis relative entropy is monotone:
$$D_q(\Phi(\rho)\|\Phi(\sigma))\leq D_q(\rho\|\sigma)\ , $$
\item For $q\in(0,2]$ in \cite{HM11, R16}, it was shown that the Tsallis relative entropy is jointly convex:
$$D_q\left(\sum_n p_n\rho_n\|\sum_np_n\sigma_n\right)\leq \sum_n p_nD_q(\rho_n\|\sigma_n)\ . $$
\item For Kraus operators $K_n$, in \cite{R16}, it was shown that the Tsallis relative entropy satisfies 
$$\sum_n D_q\left(K_n\rho K_n^\dagger\|K_n\sigma K_n^\dagger \right)\geq \sum_n p_n^q q_n^{1-q}D_q(\rho_n\|\sigma_n)\ , $$
where $p_n=\Tr(K_n\rho K_n^\dagger)$, $q_n=\Tr(K_n\sigma K_n^\dagger )$, and $\rho_n=\frac{1}{p_n}K_n\rho K_n^\dagger $, $\sigma_n=\frac{1}{q_n}K_n\sigma K_n^\dagger $.
\end{itemize}

\subsection{Tsallis Coherence}
For a fixed basis $\{\ket{j}\}_j$, the Tsallis coherence measure can be defined numerous ways. Here we present five different ways to define a Tsallis coherence measure. In this section we assume that $0<q\neq 1$.

The first way to define a coherence, is according to \cite{R16}.
\begin{definition}
For a state $\rho$, define
\begin{equation}\label{eq:Ts-1}
\cC^I_q(\rho)=\min_{\delta\in\cI} D_q(\rho\|\delta)=\frac{1}{1-q}\left\{1-\left(\sum_j \bra{j}\rho^q\ket{j}^{1/q}\right)^q\right\}\ .
\end{equation}
Here $\cI $
is the set of incoherent states in a given basis $\{\ket{j}\}_j$.
\end{definition}

The state that reaches the minimum is
$$\delta_\rho=\frac{1}{N}\sum_j\bra{j}\rho^q\ket{j}^{1/q}\ket{j}\bra{j}\ , $$
where $N=\sum_j \bra{j}\rho^q\ket{j}^{1/q}$.

This coherence measure satisfies (C1), (C2), and (C4), but not (C3).
To remedy this situation, in \cite{ZY18} authors defined a new coherence measure based on Tsallis entropy as following
\begin{definition}
For a state $\rho$, define
\begin{equation}\label{eq:Ts-2}
{\cC}^{II}_q(\rho)=\frac{1}{1-q}\left\{1-\sum_j \bra{j}\rho^q\ket{j}^{1/q} \right\}\ .
\end{equation}
\end{definition}

This coherence measure satisfies all properties (C1)-(C4), but it lacks the beauty of the distance based measure.
Straightforward observations lead to the relation
\begin{equation}\label{eq:I-II}
\cC_q^I(\rho)\leq \cC^{II}(\rho)\ . 
\end{equation}

In  \cite{Zetal17}, authors defined a coherence measure for any real symmetric concave function on a probability simplex. Such coherence satisfies (C2)-(C4) properties. And for all choices below, (C1) also holds. 

Since the function $f(p)=\frac{1}{1-q}\{\sum_jp_j^{q}-1\}$ is real symmetric and concave on a probability simplex (here $p=(p_1,\dots, p_d)$ is a probability vector), we can define a  coherence measure as follows.

\begin{definition}
For a pure state $\ket{\psi}$, define
$${\cC}^{III}_q(\psi)=\frac{1}{1-q}\left\{\sum_j|\langle{j}|{\psi}\rangle|^{2q} -1\right\}\ . $$
And for the mixed state $\rho$, define
$${\cC}^{III}_q(\rho)=\min\left\{\sum_n\lambda_n{\cC}^{III}_q(\psi_n)\ | \ \rho=\sum_n\lambda_n\ket{\psi_n}\bra{\psi_n} \right\}\ . $$
\end{definition}

Since the function $f(p)=\frac{1}{1-q}\left\{1-\left(\sum_jp_j^{1/q} \right)^q\right\} $ is real symmetric and concave for $0<q\neq 1$ and a probability  vector $p=(p_j)$, we may define a coherence measure as follows.

\begin{definition}
For any state $\rho$, define
$$ {\cC}^{IV}_q(\rho)=\min\left\{\sum_n\lambda_n{\cC}^I_q(\psi_n)\ | \ \rho=\sum_n\lambda_n\ket{\psi_n}\bra{\psi_n} \right\}\ , $$
where $\cC_q^I(\psi)$ is defined  by (\ref{eq:Ts-1}).
\end{definition}

Since the function $f(p)=\frac{1}{1-q}\{1-\sum_jp_j^{1/q}\}$ is real symmetric and concave on a probability simplex (here $p=(p_1,\dots, p_d)$ is a probability vector), we can define another coherence measure as follows.

\begin{definition}
For any state $\rho$, define
$$ {\cC}^{V}_q(\rho)=\min\left\{\sum_n\lambda_n{\cC}^{II}_q(\psi_n)\ | \ \rho=\sum_j\lambda_n\ket{\psi_n}\bra{\psi_n} \right\}\ , $$
where $\cC^{II}_q(\psi)$ is defined  by (\ref{eq:Ts-2}).
\end{definition}

Clearly, from (\ref{eq:I-II}),
$$\cC_q^{IV}(\rho)\leq \cC^V_q(\rho)\ . $$

\subsection{Discord measure}
The discord measure based on  Tsallis relative entropy is defined as following.
\begin{definition}
For $0<q\neq 1$, and any state $\rho$, define
\begin{equation}\label{eq:discord}
\cD_q(\rho)=\min_{\delta\in\cC\cC}D_q(\rho\|\delta)\ .
\end{equation}
Here $\cC\cC$ denotes the set of classically-correlated states
 $$\cC\cC=\{\rho_{AB}=\sum_{i,j}p_{i,j}\ket{ij}\bra{ij}=\sum_{i,j}p_{i,j}\ket{i}\bra{i}_A\otimes\ket{j}\bra{j}_B\ : \ \text{for any bases } \{\ket{i}_A\},  \{\ket{j}_B\}\}\ .$$
 \end{definition}
 Note that we might have written
 $$\cD_q(\rho)=\min \cC_q^I(\rho)\ , $$
 where the minimum is taken over all bases $\{\ket{i}_A\}$ and $\{\ket{j}_B\}$. Even though we will talk about other discord measure, we will not be using notation $\cD^I_q(\rho)$ for simplicity.

The Tsallis discord satisfies the properties below, which follow from the definition of the discord and the properties of Tsallis relative entropy
\begin{enumerate}
\item $\cD_q(\rho)\geq 0$, with the equality if and only if $\rho\in\cC\cC$;
\item $\cD_q(\rho_{AB})=\cD_q(U_A\otimes U_B\rho U_A^\dagger\otimes U_B^\dagger)$.
\end{enumerate}

Since
$$\cD_q(\rho)=\min \cC_q^I(\rho)\ , $$
and from Definition \ref{eq:Ts-1} of $\cC_q^I(\rho)$, we have
\begin{equation}\label{eq:D-1}
 D_q(\rho)=\min\frac{1-N_D^q(\rho)}{1-q} \ ,
 \end{equation}
 where
\begin{equation}\label{eq:N-def}
N_D(\rho):=\sum_{ij}\bra{ij}\rho^q\ket{ij}^{1/q} \ ,
\end{equation}
and the minimum is taken over all bases $\{\ket{ij}_{AB}\}_{ij}$. 

Similar to the coherence measure $\cC_q^{II}$, one may define a new discord as follows
\begin{definition}\label{def:D2}
With the above notations,
$$\cD^{II}_q(\rho)=\min_{\delta\in\cC\cC}\frac{1-(\Tr(\rho^q\delta^{1-q}))^{1/q}}{1-q}= \min\frac{1-N_D(\rho)}{1-q}=\min\frac{1}{1-q}\left\{1-\sum_{ij}\bra{ij}\rho^q\ket{ij}^{1/q}\right\}\ .$$
\end{definition}

Clearly,
$$\cD_q(\rho)\leq\cD_q^{II}(\rho)\ . $$

Using coherence $\cC^{III}_q$, we may define a new discord measure. 

\begin{definition}
Define 
\begin{equation}\label{eq:discord}
\cD_q^{III}(\rho)=\min_{\ket{ij}}\cC^{III}_q(\rho)\ ,
\end{equation}
 where the minimum is taken over all bases $\{\ket{i}_A\}$ and $\{\ket{j}_B\}$. 
\end{definition}

In Theorem \ref{lemma:1} we present explicit expressions on all these discord measures on pure states. From Theorem \ref{thm:I-III}, we obtain
$$\cD_q(\rho)\leq\cD_q^{III}(\rho)\ . $$

\subsection{Correlation measure}
\begin{definition}
The quantum correlation measure based on  Tsallis relative entropy is defined as 
\begin{equation}\label{eq:corr}
\cQ_q(\rho)=\min_{\delta\in\cC\cQ}D_q(\rho\|\delta)\ .
\end{equation}
Here $\cC\cQ$ denotes the set of classical-quantum states
 $$\cC\cQ=\{\rho_{AB}=\sum_{i}p_i\ket{i}\bra{i}\otimes \sigma_i\ | \ \text{for some basis }\{\ket{i}\} \text{ and some states }\sigma_i\}\ .$$
 \end{definition}
 
 Since $\cC\cC\subset\cC\cQ$, we have
 $$\cQ_q(\rho)\leq \cD_q(\rho)\ . $$
 
The correlation measure satisfies the following properties, \cite{L19}
\begin{enumerate}
\item $\cQ_q(\rho)\geq 0$, with the equality if and only if $\rho\in\cC\cQ$;
\item $\cQ_q(\rho_{AB})=\cQ_q(U_A\otimes U_B\rho U_A^\dagger\otimes U_B^\dagger)$;
\item $\cQ_q(\rho_{AB})\geq \cQ_q(\Phi_B(\rho_{AB}))$, where $\Phi_B$ is a local CPTP map on system $B$.
\end{enumerate}

From  \cite{L19}, for any state $\rho_{AB}$ and $q\in(0,1)\cup(1,2]$,
$$\cQ_q(\rho)=\min\frac{1}{1-q}\left\{1-\left(\sum_{ij}\bra{j}_B\bra{i}_A\rho^q \ket{i}_A^{1/q}  \ket{j}_B \right)^q \right\}\ ,$$
where minimum is taken over all bases $\{\ket{i}_A\}_i$ on system $A$.
Denote 
\begin{equation}\label{eq:N_Q}
{N}_Q(\rho)=\sum_{ij}\bra{j}_B\bra{i}_A\rho^q \ket{i}_A^{1/q}  \ket{j}_B\ . 
\end{equation}
Therefore,
\begin{equation}\label{eq:Q}
\cQ_q(\rho)=\min\frac{1-N_Q^q(\rho)}{1-q}\ .
\end{equation}

\begin{lemma}\label{lemma:2}
Let $\ket{\psi_k}$ be the diagonal basis for $\rho$, i.e. let
$$\rho=\sum_k\lambda_k\ket{\psi_k}\bra{\psi_k}\ , $$
where $\lambda_k\geq 0$, and $\sum_k\lambda_k=1$. Also, let $\ket{\xi_{kn}}_A$ and $\ket{\xi_{kn}}_B$ be the Hilbert-Schmidt basis for the state $\ket{\psi_k}$, i.e.
\begin{equation}\label{eq:psi}
\ket{\psi_k}=\sum_n\alpha_{kn}\ket{\xi_{kn}\xi_{kn}}_{AB}\ . 
\end{equation}
Here $\alpha_{kn}\geq 0$ and for every $k$, and $\sum_n\alpha_{kn}^2=1$. 
Then
\begin{equation}\label{eq:N-Q}
N_Q(\rho)=\sum_{ik}\lambda_k\left(\sum_n\alpha_{kn}^2|\langle{i}|\xi_{kn}\rangle|^2 \right)^{1/q} \ .
\end{equation}
\end{lemma}
\begin{proof}
By definition (\ref{eq:N_Q}), we have
$${N}_Q(\rho)=\sum_{ij}\bra{j}_B\left(\sum_k\lambda_k^q\bra{i}_A\ket{\psi_k}\bra{\psi_k}\ket{i}_A \right)^{1/q} \ket{j}_B\ . $$
Denote $$\ket{\Psi_{ik}}_B:=\bra{i}_A\ket{\psi_k}_{AB}=\sum_n\alpha_{kn}\langle{i}|{\xi_{kn}}\rangle\ket{\xi_{kn}}_B=:\sqrt{m_{ik}}\ket{\phi_{ik}}\ ,$$
where
$$m_{ik}=\sum_n\alpha_{kn}^2|\langle{i}|{\xi_{kn}}\rangle|^2, \ \text{ and }\ \langle{\phi_{ik}}|{\phi_{ik}}\rangle=1 \ .$$
Then
$$N_Q(\rho)=\sum_{ij}\bra{j}\left( \sum_k\lambda_k^q m_{ik}\ket{\phi_{ik}}\bra{\phi_{ik}}\right)^{1/q} \ket{j}=\sum_{ijk}\lambda_k m_{ik}^{1/q}|\langle{j}|{\phi_{ik}}\rangle|^2=\sum_{ik}\lambda_k m_{ik}^{1/q}\ .$$
By definition of $m_{ik}$, we obtain the statement of the lemma.
\end{proof}

\begin{definition}\label{def:Q2}
In  \cite{L19}, authors defined a new correlations measure, 
$$\cQ^{II}_q(\rho)=\min\frac{1-N_Q(\rho)}{1-q}=\min\frac{1}{1-q}\left\{1-\sum_{ij}\bra{j}_B\bra{i}_A\rho^q \ket{i}_A^{1/q}  \ket{j}_B \right\}\ .$$
\end{definition}

\begin{lemma}\label{lemma:QD2}
For any state $\rho$,
$$\cQ_q^{II}(\rho)\leq \cD_q^{II}(\rho)\ . $$
\end{lemma}
\begin{proof}
Since $x^q$ is monotone, the minimum for both $\cQ_q$ and $\cQ^{II}_q$ is achieved on the same basis. Similarly, from  the expression (\ref{eq:D-1}) and Definition \ref{def:D2}, the minimum for both  $\cD_q$ and $\cD^{II}_q$ is achieved on the same basis. Since $\cQ_q(\rho)\leq \cD_q(\rho) $, the Lemma follows.
\end{proof}

\section{Relations between coherence measures}\label{sec:coherence}

Here we compare all five coherence measure with each other.

\begin{theorem}\label{thm:I-III}
For $q\in(0,1)$, 
$$\cC_q^I(\psi)\leq {\cC}^{III}_q(\psi)\ . $$
And clearly, therefore, ${\cC}^{IV}_q(\rho)\leq  {\cC}^{III}_q(\rho).$
\end{theorem}
\begin{proof}
The difference is
$$\widehat{\cC}_q(\psi)-\cC_q(\psi)=\frac{1}{1-q}\left\{\sum_j |\bra{j}\ket{\psi}|^{2q}+\left(\sum_j|\bra{j}\ket{\psi}|^{2/q}\right)^q-2\right\}\ .$$
The expression in the parenthesis is monotone decreasing in $q\in(0,1)$, which can easily be seen by taking the derivate with respect to $q$. Therefore, the minimal value of the expression in the parenthesis occurs at $q=1$, which is exactly $0$.
\end{proof}

\begin{theorem}\label{thm:I-IV}
For $0<q\neq 1$,
$$\cC_q^I(\rho)\leq \cC_q^{IV}(\rho)\ . $$
\end{theorem}
\begin{proof}
Let $\rho=\sum_n\lambda_n\ket{\psi_n}\bra{\psi_n}$ be the optimal decomposition for $\cC^{IV}_q(\rho)$. Then
$$\cC^{IV}_q(\rho)=\frac{1}{1-q}\left\{1-\sum_k\lambda_k\left(\sum_j|\bra{j}\ket{\psi_k}|^{2/q} \right)^q \right\}\ , $$
and
$$\cC_q^I(\rho)=\frac{1}{1-q}\left\{1-\left(\sum_{j}\left( \sum_k\lambda_k^q |\bra{j}\ket{\psi_k}|^{2}\right)^{1/q}\right)^q\right\}\ . $$
For $q<1$, 
$$\left( \sum_k\lambda_k^q |\bra{j}\ket{\psi_k}|^{2}\right)^{1/q}\geq \sum_k\lambda_k |\bra{j}\ket{\psi_k}|^{2/q}\ .$$
For $q>1$, the inequality is reversed. Therefore,
$$\cC_q^I(\rho)=\frac{1}{1-q}\left\{1-\left(\sum_{jk}\lambda_k |\bra{j}\ket{\psi_k}|^{2/q}\right)^q\right\}\ . $$

The statement follows from convexity of the function $\frac{1}{1-q}\{1-x^q\}$.
\end{proof}

\begin{theorem}\label{thm:II-V}
For $0<q\neq 1$,
$$\cC_q^{II}(\rho)\leq \cC_q^{V}(\rho)\ . $$
\end{theorem}
\begin{proof}
Let $\rho=\sum_n\lambda_n\ket{\psi_n}\bra{\psi_n}$ be the optimal decomposition for $\cC^{V}_q(\rho)$. Then
$$\cC^V_q(\rho)=\frac{1}{1-q}\left\{1-\sum_k\lambda_k \sum_j|\bra{j}\ket{\psi_k}|^{2/q} \right\}\ . $$
Also,
$$\cC^{II}_q(\rho)=\frac{1}{1-q}\left\{1-\sum_{j}\left( \sum_k\lambda_k^q |\bra{j}\ket{\psi_k}|^{2}\right)^{1/q}\right\}\ . $$
For $q<1$, 
$$\left( \sum_k\lambda_k^q |\bra{j}\ket{\psi_k}|^{2}\right)^{1/q}\geq \sum_k\lambda_k |\bra{j}\ket{\psi_k}|^{2/q}\ .$$
For $q>1$, the inequality is reversed.
\end{proof}

\begin{remark}
Note that there is no straight forward comparison between $\cC^{II}(\psi)$ and $\cC^{III}(\psi)$, as one is not majorized by the other for all cases. For example, in $d=2$, let us denote $x:=|\bra{1}\ket{\psi}|^2$. Then
$$\text{if } x=1/2, q=1/2,\ \text{then } \cC^{II}_{1/2}(\psi)\geq \cC^{III}_{1/2}(\psi)\ , $$
but 
$$\text{if } x=0.1, q=1/2,\ \text{then } \cC^{II}_{1/2}(\psi)\leq \cC^{III}_{1/2}(\psi)\ .$$
\end{remark}

In conclusion, here are the inequalities that we have proved:
$$\cC^I_q(\rho)\leq \cC^{IV}_q(\rho)\leq \cC^{III}_q(\rho) \ ,$$
$$ \cC^{IV}_q(\rho)\leq  \cC^{V}_q(\rho)\ , $$
$$ \cC^{I}_q(\rho)\leq  \cC^{II}_q(\rho)\leq  \cC^{V}_q(\rho)\ . $$

\section{Discord and correlation measure on pure states}

For pure states, we can explicitly calculate the minimum in both expressions for correlation and discord measures. This was proved for the correlation measure $\cQ_q$ in  \cite{L19}. 

\begin{theorem}\label{lemma:1}
Let $\rho=\ket{\psi}\bra{\psi}$ be a pure state. Let $\{\alpha_n\}$ be the Hilbert-Schmidt coefficients of $\ket{\psi}_{AB}=\sum_n\alpha_n\ket{\xi_n}_A\ket{\xi_n}_B\ .$ For $q\in(0,1)\cup(1,2]$, the discord and correlation measure can be explicitly calculated
\begin{equation}\label{eq:D}
\cD_q(\ket{\psi}\bra{\psi})=\cQ_q(\ket{\psi}\bra{\psi})=\frac{1}{1-q}\left\{1-\left(\sum_{n}\alpha_{n}^{2/q}\right)^q\right\}\geq 0\ . 
\end{equation}
\begin{equation}\label{eq:D-II}
\cD_q^{II}(\ket{\psi}\bra{\psi})=\cQ^{II}_q(\ket{\psi}\bra{\psi})=\frac{1}{1-q}\left\{1-\sum_{n}\alpha_{n}^{2/q}\right\}\geq 0\ . 
\end{equation}
\begin{equation}\label{eq:D-III}
\cD_q^{III}(\ket{\psi}\bra{\psi})=\frac{1}{1-q}\left\{\sum_{n}\alpha_{n}^{2q}-1\right\}\geq 0\ . 
\end{equation}
\end{theorem}

\begin{proof}

Since the expression (\ref{eq:D}) was proved for the correlation measure in  \cite{L19}, and
\begin{equation}\label{eq:QD}
\cQ_q(\rho)\leq \cD_q(\rho)\ , 
\end{equation}
we need to show that for some particular basis, the discord has the same expression. By definition of $N_D$ (\ref{eq:N-def}),  
$$N_D(\psi)=\sum_{ij}|\bra{ij}\ket{\psi}|^{2/q}=\sum_{ij}\left|\sum_n \alpha_{n}\bra{i}\ket{\xi_{n}}\bra{j}\ket{\xi_{n}} \right|^{2/q}\ .$$
 Choose the orthonormal basis $\ket{i}_A:=\ket{\xi_i}_A$ on system $A$ and  $\ket{j}_B:=\ket{\xi_j}_B$ on system $B$. Then, we have 
 \begin{equation}\label{eq:N-a}
 N_D(\psi)=  \sum_{n}  \alpha_{n}^{2/q}\ .
 \end{equation}
 By (\ref{eq:D-1}), this leads to the first statement of the Theorem.
 
For the next two expressions, we recall the proof of (\ref{eq:QD}) for pure states, and modify it later as necessary for other discord and correlation measures.
From (\ref{eq:N-def}) for a pure state
$$N_D(\psi)=\sum_{ij}(\bra{ij}\ket{\psi}\bra{\psi}\ket{ij})^{1/q}=:\sum_{ij}\bra{j} R_i \ket{j}^{1/q}\ , $$
where
$$R_i=\left(\bra{i}\otimes I\right) \ket{\psi}\bra{\psi} \left(\ket{i}\otimes I\right)\ . $$
Clearly, $R_i$ is a projector in system $B$. Denote
$$m_i:=\sum_n\alpha_n^2|\bra{i}\ket{\xi_n}|^2 \ ,
\text{ and }\ 
\ket{\phi_i}:=\frac{1}{\sqrt{m_i}} \bra{i}\otimes I \ket{\psi}\ .$$
Note that $\bra{\phi_i}\ket{\phi_i}=1$ for every $i$. Then
$$R_i=m_i\ket{\phi_i}\bra{\phi_i}\ . $$
Therefore,
$$N_D(\psi)=\sum_i m_i^{1/q}\sum_j |\bra{j}\ket{\phi_i}|^{2/q}\ .$$
For $q\in(0,1)$, for every $j$, $|\bra{j}\ket{\phi_i}|^{2/q}\leq |\bra{j}\ket{\phi_i}|^{2}$. Therefore from Lemma \ref{lemma:2} we obtain the last equality
$$N_D(\psi)\leq  \sum_i m_i^{1/q}\sum_j |\bra{j}\ket{\phi_i}|^{2}=\sum_i m_i^{1/q}=\sum_i \left( \sum_n\alpha_n^2|\bra{i}\ket{\xi_n}|^2\right)^{1/q}=N_Q(\psi)\ .$$
For $q\in(0,1)$, the function $x^{1/q}$ is convex. Therefore,  
\begin{equation}\label{eq:NDQ}
N_D(\psi)\leq N_Q(\psi)\leq \sum_i\sum_n\alpha_n^{2/q}|\bra{i}\ket{\xi_n}|^2= \sum_n\alpha_n^{2/q}\ .
\end{equation}
For $q>1$, the above inequality is reversed. Since $x^q$ is monotone increasing for all $q$, and together with the pre-factor of $\frac{1}{1-q}$, we obtain (\ref{eq:QD}). 

For the second inequality recall that
$$\cD^{II}_q(\rho)=\min\frac{1-N_D(\rho)}{1-q}\ , $$
and $$\cQ^{II}_q(\rho)=\min\frac{1-N_Q(\rho)}{1-q}\ .$$
so the proof just carries as is from (\ref{eq:N-a}) and (\ref{eq:NDQ}).

For the third inequality, recall that 
$$\cD_q^{III}(\psi)=\min_{\ket{ij}}\cC^{III}_q(\psi)=\min \frac{1}{1-q}\left\{\sum_{ij}|\langle{ij}|{\psi}\rangle|^{2q} -1\right\}\ .$$
Since for all $q\in(0,1)$ we have $N_D(\psi)=\sum_{ij}|\langle{ij}|{\psi}\rangle|^{2/q}\leq \sum_n\alpha_n^{2/q}$, and for $q>1$, $N_D(\psi)\geq \sum_n\alpha_n^{2/q}$, we obtain for $q\in(0,1)\cup(1,2]$
$$\frac{1}{1-q}\left\{\sum_{ij}|\langle{ij}|{\psi}\rangle|^{2q}-1\right\}\geq \frac{1}{1-q}\left\{\sum_n\alpha_n^{2q}-1 \right\}\ . $$

For all the cases, the inequality is reached when the basis is taken to be the Hilbert-Schmidt basis of state $\ket{\psi}$.

\end{proof}

\section{Bounds on quantum discord and correlation measure}

Here we present the tight lower and upper bounds for a quantum discord for mixed states.

\begin{theorem}\label{thm:lower}
For a mixed state $\rho$ and $q\in(0,1)\cup(1,2]$, the discord and correlation measure are tightly lower bounded by
\begin{equation}\label{eq:lower}
\cD_q(\rho)\geq\cQ_q(\rho)\geq \frac{1}{1-q}\left\{1-\left(\sum_{kn}\lambda_k\alpha_{kn}^{2/q} \right)^q \right\}\ ,
\end{equation}
\begin{equation}\label{eq:lower2}
\cD^{II}_q(\rho)\geq\cQ^{II}_q(\rho)\geq \frac{1}{1-q}\left\{1-\sum_{kn}\lambda_k\alpha_{kn}^{2/q} \right\}\ .
\end{equation}
Both inequalities are saturated on any pure state, where we recover the equality (\ref{eq:D}).
\end{theorem}
\begin{proof}
From Lemma \ref{lemma:2},
$$N_Q(\rho)=\sum_{ik}\lambda_k\left(\sum_n\alpha_{kn}^2|\bra{i}\ket{\xi_{kn}}|^2 \right)^{1/q}=\sum_k\lambda_kN_Q({\psi_k}) \ .
 $$
 By (\ref{eq:NDQ}) for $q\in(0,1)$, we have
  $$N_Q(\rho)\leq \sum_{kn}\lambda_k\alpha_{kn}^{2/q}\ . $$
  From (\ref{eq:Q}) and Definition \ref{def:Q2}, 
  $$\cQ_q(\rho)=\min\frac{1-N_Q^q(\rho)}{1-q}\ , $$
  and 
  $$\cQ^{II}_q(\rho)=\min\frac{1-N_Q(\rho)}{1-q}\ .$$
  we obtain the statement of the Theorem for both correlation measures. Similarly, the inequality holds for $q\in(1,2]$. From Lemma \ref{lemma:QD2} and (\ref{eq:QD}), the inequalities for discord hold.

\end{proof}

The following upper bound was derived  in \cite{R16} for the coherence measure $\cQ_q$, which gives the upper bound for the discord $\cD_q$ automatically. Here we provide the full proof of that part for completeness, and show that the upper bound is tight for the discord and correlation measures.

\begin{theorem}\label{thm:upper}
For a mixed state $\rho$, and $q\in(0,1)\cup(1,2]$, the discord and correlation measure are tightly upper bounded by
\begin{equation}\label{eq:upper}
\cQ_q(\rho)\leq\cD_q(\rho)\leq \frac{1}{1-q}\left\{1-d^{2(q-1)} \Tr\rho^q \right\}\leq  \frac{1}{1-q}\left\{1-d^{2(q-1)} \right\}\ ,
\end{equation}
\begin{equation}\label{eq:upper2}
\cQ^{II}_q(\rho)\leq\cD^{II}_q(\rho)\leq \frac{1}{1-q}\left\{1-d^{2(q-1)/q} (\Tr\rho^q)^{1/q} \right\}\leq  \frac{1}{1-q}\left\{1-d^{2(q-1)/q} \right\}\ .
\end{equation}
Here $d_A=d_B=d$.
All inequalities are  saturated for a pure state $\rho=\ket{\psi}\bra{\psi}$ with Hilbert-Schmidt coefficients of $1/d$, i.e. there is a basis $\{\ket{\xi_i\xi_j}_{AB}\}$ such that
$$\ket{\psi}=\sum_{ij}\frac{1}{d}\ket{\xi_i\xi_j}\ . $$
\end{theorem}
\begin{proof}
Let $\delta=\frac{1}{d^2}I$ be a completely mixed  state. Then
$$D_q(\rho\|\delta)=\frac{1}{1-q}\left\{1-d^{2(q-1)} \Tr\rho^q \right\}\geq \cD_q(\rho)\ . $$
From Lemma \ref{lemma:1}, the inequality is achieved for the pure state $\rho=\ket{\psi}\bra{\psi}$, with Hilbert-Schmidt coefficients of $1/d$, i.e. there is a basis $\{\ket{\xi_i\xi_j}_{AB}\}$ such that
$$\ket{\psi}=\sum_{ij}\frac{1}{d}\ket{\xi_i\xi_j}\ . $$

The second inequality holds since $\Tr\rho^q\geq \Tr\rho=1$ for $q\in(0,1)$, and $\Tr\rho^q\leq \Tr\rho=1$ for $q>1$, with equality on any pure state. 

Since the minimum for both  $\cD_q$ and $\cD^{II}_q$ is achieved on the same basis, the inequality for $\cD^{II}_q$ follows. From Lemma \ref{lemma:QD2} and (\ref{eq:QD}), the inequalities for correlation measures hold.
\end{proof}

\begin{remark}
Consider a pure state $\rho=\ket{\psi}\bra{\psi}$ with
$\ket{\psi}=\sum_{ij}\frac{1}{d}\ket{ij}, $
where $\{\ket{ij}_{AB}\}_{ij}$ being any basis in systems $A$ and $B$. For this state the equality (\ref{eq:D}), the lower (\ref{eq:lower}) and the upper (\ref{eq:upper}) bounds become the same and equal to 
$$\cQ_q(\ket{\psi}\bra{\psi})=\cD_q(\ket{\psi}\bra{\psi})=\frac{1}{1-q}\left\{1- d^{2(q-1)} \right\}\ . $$
Moreover, for the fixed bases $\{\ket{i}_A\}_i$ and $\{\ket{j}_B\}_j$ on systems $A$ and $B$, used in the coherence, from (\ref{eq:Ts-1}), the maximally mixed state $\ket{\psi}=\sum_{ij}\frac{1}{d}\ket{ij}$,  has equal coherence, discord, and correlation measure
$$\cC^I_q(\ket{\psi}\bra{\psi})=\frac{1}{1-q}\left\{1- d^{2(q-1)} \right\}=\cD_q(\ket{\psi}\bra{\psi})=\cQ_q(\ket{\psi}\bra{\psi})\ . $$
Similar expressions hold for $\cC^{II}_q$, $\cD^{II}_q$ and $\cQ_q^{II}$.

\end{remark}

\vspace{0.3in}
\textbf{Acknowledgments.}  A. V. is supported by NSF grant DMS-1812734.

\end{document}